\documentclass[letterpaper, 10pt, conference]{ieeeconf}
\IEEEoverridecommandlockouts
\usepackage{amsmath,amssymb}
\usepackage{graphicx}
\usepackage{cite}
\usepackage{algorithmic}
\usepackage{algorithm}
\usepackage{booktabs}
\usepackage{tikz}
\usepackage{verbatim}
\usepackage{url}
\usepackage{subcaption}
\usepackage{bbm}
\usepackage{hyperref}

\PassOptionsToPackage{dvipsnames}{xcolor}

\newtheorem{theorem}{Theorem}

\newtheorem{assumption}{Assumption}
\newtheorem{remark}{Remark}
\newtheorem{definition}{Definition}

\title{\LARGE \bf
  CHMAS: A Coupled Hierarchical Framework for Multi-Agent Reinforcement Learning
}

\author{Dongming Wang, Jie Xu, Yanyu Zhang, and Wei Ren
	\thanks{This work was supported by the National Science Foundation under Grant ECCS-2129949.}
    \thanks{The authors are with the Department of Electrical and
    Computer Engineering, University of California, Riverside,
    CA 92521, USA.
    {\tt\small \{wdong025, jxu150, yzhan831, wei.ren\}@ucr.edu}.}%
}

\begin{document}

\maketitle
\thispagestyle{empty}
\pagestyle{empty}

%%%%%%%%%%%%%%%%%%%%%%%%%%%%%%%%%%%%%%%%%%%%%%%%%%%%%%%%%%%%%%%%
\begin{abstract}
  Multi-agent reinforcement learning (MARL) systems face fundamental
  challenges in balancing global coordination with local execution
  across different temporal scales. This paper introduces the Coupled
  Hierarchical Multi-Agent System (CHMAS), a novel framework that
  decomposes multi-agent decision-making into centralized strategic
  planning and distributed tactical execution with bidirectional
  information flow. The strategic layer integrates all agents' states
  with an exclusive global environmental state to generate guidance
  actions every $T$ timesteps, while tactical agents execute
  distributed policies augmented by strategic guidance and local
  neighborhood observations. Unlike existing hierarchical approaches
  with unidirectional control, CHMAS establishes a feedback mechanism
  where accumulated tactical rewards influence strategic objectives
  through a coupling coefficient $\lambda$, ensuring strategic plans
  remain grounded in tactical feasibility. To address the
  non-stationarity inherent in hierarchical learning, we propose an
  asynchronous update protocol where strategic parameters update every
  $N_f$ tactical episodes, allowing tactical policies to converge to
  quasi-stationary points between strategic changes. We present both a
  general bi-level formulation capturing full system dynamics and a
  tractable additive approximation enabling rigorous analysis.
  Theoretical analysis proves that this asynchronous scheme achieves
  $\mathcal{O}(\log K/\sqrt{K})$ convergence for the strategic layer
  after $K$ strategic updates under standard assumptions. Experimental
  validation in a multi-agent foraging domain demonstrates successful
  learning of spatially partitioned exploration strategies, with both
  layers converging stably despite hierarchical coupling.
\end{abstract}

%%%%%%%%%%%%%%%%%%%%%%%%%%%%%%%%%%%%%%%%%%%%%%%%%%%%%%%%%%%%%%%%
\section{INTRODUCTION}

Multi-agent reinforcement learning (MARL) has emerged as a fundamental
paradigm for autonomous coordination in complex environments, with
applications spanning robotics \cite{nguyen2020deep}, traffic management
\cite{wiering2000multi}, and resource allocation \cite{qu2020scalable}.
Despite significant progress, real-world deployment faces a critical
challenge: coordinating distributed agents across multiple temporal and
spatial scales while maintaining computational tractability.

The curse of dimensionality renders purely centralized approaches
intractable. For $N$ agents with state spaces $\mathcal{S}_i$ and
action spaces $\mathcal{A}_i$, the joint space grows as
$O(\prod_{i=1}^N |\mathcal{S}_i| \cdot |\mathcal{A}_i|)$
\cite{zhang2021multi}. While distributed methods address scalability
through local decision-making \cite{zhang2018fully}, they sacrifice
coordination quality for long-horizon objectives. Traditional approaches
like QMIX \cite{rashid2020monotonic} and MADDPG \cite{lowe2017multi}
leverage centralized training with decentralized execution, but assume
all agents access relevant global information during training, which
often fails in practice due to partial observability and communication
constraints \cite{hernandez2019survey}.

Furthermore, multi-agent systems inherently operate across multiple
temporal scales. Strategic decisions (formation control, task
allocation) require global information including environmental states
inaccessible to individual agents, operating on extended timescales.
Tactical responses (collision avoidance, local adaptation) demand
rapid reactions based on immediate observations. While hierarchical
reinforcement learning offers a natural framework for multi-scale
decision-making \cite{sutton1999between, bacon2017option}, existing
hierarchical MARL approaches \cite{pateria2021hierarchical,
vezhnevets2017feudal} employ unidirectional control where high-level
policies dictate low-level actions, missing opportunities for tactical
feedback to inform strategic planning.

The interaction between hierarchical levels creates complex learning
dynamics. Strategic changes fundamentally alter the environment
perceived by tactical agents, violating stationarity assumptions
\cite{hernandez2019survey}. Conversely, evolving tactical policies
change execution characteristics that strategic planners rely upon.
Previous hierarchical MARL work \cite{kulkarni2016hierarchical,
nachum2018data} addresses this through reward engineering or option
learning but requires domain expertise and often fails to generalize.

This paper introduces CHMAS, which addresses these limitations through
bidirectional coupling between strategic and tactical layers at
different temporal scales. The strategic level maintains centralized
access to both agents' states and an exclusive global environmental
state $\mathcal{S}^{\text{env}}$ capturing system-wide properties
essential for coordination. Strategic guidance, generated every $T$
timesteps, broadcasts to tactical agents who execute distributed
policies based on augmented local observations. Tactical rewards
accumulate and influence strategic objectives through coupling
coefficient $\lambda$, creating feedback that grounds strategic
planning in tactical realities. An asynchronous protocol with
strategic updates every $N_f$ tactical episodes manages non-stationary
dynamics, allowing tactical convergence between strategic changes.

Our contributions are threefold. First, we formalize the hierarchical
MARL problem with bidirectional coupling and information asymmetry,
providing both a general bi-level formulation and tractable additive
approximation. Second, we develop the Asynchronous Hierarchical Policy
Gradient (AHPG) algorithm managing non-stationarity through temporal
separation while maintaining $\mathcal{O}(\log K/\sqrt{K})$
convergence guarantees for the strategic layer. Third, we validate the
framework empirically, demonstrating stable bidirectional learning in a
cooperative multi-agent foraging task.

%%%%%%%%%%%%%%%%%%%%%%%%%%%%%%%%%%%%%%%%%%%%%%%%%%%%%%%%%%%%%%%%
\section{PRELIMINARIES}

We consider a cooperative multi-agent system with $N$ agents operating
in a shared environment. The system employs a hierarchical architecture
with a strategic layer that accesses global information at low temporal
frequency and a tactical layer that operates distributedly at high
frequency using only local information.

\subsection{Multi-Agent Communication Graph}

The tactical agents interact through a time-invariant communication
graph $\mathcal{G} = (\mathcal{V}, \mathcal{E})$ where
$\mathcal{V} = \{1, 2, \ldots, N\}$ represents the agent set and
$\mathcal{E} \subseteq \mathcal{V} \times \mathcal{V}$ denotes the
edge set \cite{olfati2007consensus,nedic2010constrained}. For each
agent $i$, the neighborhood $\mathcal{N}_i = \{j \in \mathcal{V} :
(i,j) \in \mathcal{E}\}$ is the set of agents that can directly
communicate with agent $i$. This graph constrains information flow at
the tactical level, ensuring distributed execution remains scalable.

\subsection{Strategic MDP}

\begin{definition}[Strategic MDP]
  The strategic layer evolves as
  $\mathcal{M}^{\text{str}} = (\mathcal{S}^{\text{str}},
  \mathcal{A}^{\text{str}}, P^{\text{str}}, R^{\text{str}},
  \gamma_{\text{str}})$, where:
  (i)~the state $\mathcal{S}^{\text{str}} = \mathcal{S}^{\text{env}}
  \times \prod_{i=1}^N \mathcal{S}_i$ combines the exclusive global
  environmental state $\mathcal{S}^{\text{env}}$ (resource
  distributions, global objectives) with all agent states, so that at
  epoch $k$ the strategic state is
  $s^{\text{str}}_k = (s^{\text{env}}_k, s_{1,k}, \ldots, s_{N,k})$;
  (ii)~actions $\mathcal{A}^{\text{str}} = \prod_{i=1}^N
  \mathcal{A}^{\text{str}}_i$ provide personalized guidance
  $a^{\text{str}}_{i,k}$ to each agent, persisting for the interval
  $[kT,(k+1)T)$; and (iii)~the reward integrates global objectives
  with tactical feedback:
  \begin{equation}
    \label{eq:strategic_reward}
    R^{\text{str}}(s^{\text{str}}_k, \mathbf{a}^{\text{str}}_k)
    = R^{\text{global}}(s^{\text{env}}_k, \mathbf{a}^{\text{str}}_k)
    + \lambda \sum_{t=kT}^{(k+1)T-1} \sum_{i=1}^N r^{\text{tac}}_{i,t},
  \end{equation}
  where $\lambda \in [0,1]$ controls the influence of tactical
  execution quality on strategic evaluation.
\end{definition}

The strategic policy $\pi^{\text{str}}: \mathcal{S}^{\text{str}}
\rightarrow \mathcal{P}(\mathcal{A}^{\text{str}})$ maps complete state
observations to guidance distributions, encoding high-level
coordination logic that individual agents cannot compute from local
information alone.

\subsection{Tactical Multi-Agent MDP}

\begin{definition}[Tactical MAMDP]
  Given fixed strategic guidance $\mathbf{a}^{\text{str}}_k$, the
  tactical layer evolves as
  $\mathcal{M}^{\text{tac}} = (\{\mathcal{S}_i\},
  \{\mathcal{O}^{\text{tac}}_i\}, \{\mathcal{A}^{\text{tac}}_i\},
  P^{\text{tac}}, \{R^{\text{tac}}_i\}, \gamma_{\text{tac}})$, where:
  (i)~each agent $i$ observes
  \begin{equation}
    \label{eq:tactical_obs}
    o^{\text{tac}}_{i,t}
    = \phi_i(s_{i,t}, \{s_{j,t}\}_{j \in \mathcal{N}_i},
      a^{\text{str}}_{i,k}) \in \mathcal{O}^{\text{tac}}_i,
  \end{equation}
  combining local state, neighbor states, and strategic guidance
  ($k = \lfloor t/T \rfloor$);
  (ii)~the joint transition factorizes over the communication graph:
  $P^{\text{tac}}(\mathbf{s}_{t+1} | \mathbf{s}_t,
  \mathbf{a}^{\text{tac}}_t) = \prod_{i=1}^N
  P_i(s_{i,t+1} | s_{i,t},
  \{s_{j,t}, a^{\text{tac}}_{j,t}\}_{j \in \mathcal{N}_i \cup \{i\}})$,
  maintaining tractability as $N$ grows \cite{zhang2021multi}; and
  (iii)~the per-agent reward
  \begin{equation}
    \label{eq:tactical_reward}
    R^{\text{tac}}_i(s_{i,t}, a^{\text{tac}}_{i,t},
    \{s_{j,t}, a^{\text{tac}}_{j,t}\}_{j \in \mathcal{N}_i},
    a^{\text{str}}_{i,k})
  \end{equation}
  enables local optimization and neighbor coordination while
  incorporating strategic objectives.
\end{definition}

\subsection{Bidirectional Coupling}

Unlike traditional top-down hierarchies \cite{pateria2021hierarchical},
CHMAS establishes feedback in both directions. \emph{Downward coupling}
(strategic $\to$ tactical): guidance $a^{\text{str}}_{i,k}$ shapes
tactical behavior through observation augmentation
\eqref{eq:tactical_obs}, reward modulation \eqref{eq:tactical_reward},
and temporal persistence over $T$ timesteps. \emph{Upward coupling}
(tactical $\to$ strategic): accumulated rewards
$R^{\text{tac,cum}}_k = \sum_{t=kT}^{(k+1)T-1} \sum_{i=1}^N
r^{\text{tac}}_{i,t}$ influence strategic rewards via $\lambda$ in
\eqref{eq:strategic_reward}, so that strategic decisions account for
tactical feasibility. Setting $\lambda = 0$ recovers a purely top-down
hierarchy. The framework maintains asymmetric information access:
the strategic layer observes complete state $s^{\text{str}}_k$ while
tactical agents access only $\mathcal{N}_i$, with strategic guidance
serving as an information bottleneck \cite{tishby2000information}.

\section{PROBLEM FORMULATION}\label{sec:pf}

\subsection{General Coupled Formulation}

The bidirectional coupling creates a bi-level optimization structure
\cite{colson2007overview,sinha2018review}. The strategic objective:
\begin{equation}
  \label{eq:strategic_objective}
  J^{\text{str}}(\theta^{\text{str}} | \{\theta^{\text{tac}}_i\})
  = \mathbb{E}\!\left[\sum_{k=0}^{\infty} \gamma^k_{\text{str}}
    R^{\text{str}}(s^{\text{str}}_k, \mathbf{a}^{\text{str}}_k)\right]
\end{equation}
depends on tactical policies through \eqref{eq:strategic_reward}.
Each tactical objective:
\begin{equation}
  \label{eq:tactical_objective}
  J^{\text{tac}}_i(\theta^{\text{tac}}_i | \theta^{\text{str}})
  = \mathbb{E}\!\left[\sum_{t=0}^{\infty} \gamma^t_{\text{tac}}
    r^{\text{tac}}_{i,t}\right]
\end{equation}
depends on strategic parameters through \eqref{eq:tactical_obs} and
\eqref{eq:tactical_reward}.

\begin{definition}[Coupled Hierarchical Problem]
  The coupled optimization problem is the bi-level program:
  \begin{align*}
    \max_{\theta^{\text{str}}} \quad
    & J^{\text{str}}(\theta^{\text{str}} |
      \{\theta^{\text{tac}*}_i(\theta^{\text{str}})\}) \\
    \text{s.t.} \quad
    & \theta^{\text{tac}*}_i(\theta^{\text{str}})
    = \arg\max_{\theta^{\text{tac}}_i}
      J^{\text{tac}}_i(\theta^{\text{tac}}_i | \theta^{\text{str}}),
      \quad \forall i \in \mathcal{V},
  \end{align*}
  where the strategic optimization accounts for tactical best responses.
\end{definition}

\begin{remark}
  We operate in the cooperative MARL setting where agents share a
  common team objective. Accordingly, the per-agent best responses
  $\theta^{\text{tac}*}_i$ are aligned with the joint team optimum,
  and no equilibrium-seeking across agents is required at the tactical
  level \cite{zhang2021multi}.
\end{remark}

This Stackelberg structure is computationally intractable in general
\cite{hong2023two}, motivating the approximation below.

\subsection{Additive Approximation for Tractable Optimization}

\begin{definition}[Additive Hierarchical Problem]
  Under the additive approximation, the global objective is:
  \begin{equation}
    \label{eq:additive}
    J^{\text{global}}(\theta^{\text{str}}, \{\theta^{\text{tac}}_i\})
    = J^{\text{str}}(\theta^{\text{str}})
    + \sum_{i=1}^N J^{\text{tac}}_i(\theta^{\text{tac}}_i),
  \end{equation}
  where each term is evaluated with fixed parameters of the other layer.
\end{definition}

\begin{remark}
  Equation~\eqref{eq:additive} is an approximation of the true coupled
  objective: the cross-layer dependencies in
  \eqref{eq:strategic_objective}--\eqref{eq:tactical_objective} are not
  eliminated but handled implicitly through the asynchronous update
  schedule. The approximation error is controlled by the tactical
  optimality gap $\delta_k$ which will be defined in Assumption~\ref{assum:tac} at each epoch,
  which vanishes as learning progresses.
\end{remark}

The separable gradient structure,
\begin{equation}
  \label{eq:gradient}
  \nabla_{\theta} J^{\text{global}}
  = \bigl[\nabla_{\theta^{\text{str}}} J^{\text{str}},\;
    \nabla_{\theta^{\text{tac}}} J^{\text{tac}}\bigr]^T,
\end{equation}
allows updates to be treated as coordinate steps in the joint parameter
space, enabling rigorous convergence analysis via two-timescale
stochastic approximation \cite{borkar1997stochastic,
nesterov2012efficiency}.

%%%%%%%%%%%%%%%%%%%%%%%%%%%%%%%%%%%%%%%%%%%%%%%%%%%%%%%%%%%%%%%%
\section{ASYNCHRONOUS LEARNING ALGORITHM}

\subsection{Algorithm Design Principles}

AHPG leverages three principles for stable hierarchical learning.

\emph{Temporal separation.} Strategic decisions occur every $T$
timesteps and persist throughout the interval. Strategic parameters
follow a decaying step size schedule $\eta^{\text{str}}_k =
\alpha/\sqrt{k}$ and update only every $N_f$ tactical episodes,
operating at a slower effective timescale than tactical learning. This
allows tactical policies to converge to near-optimal responses between
strategic changes, transforming the non-stationary hierarchical problem
into alternating stationary subproblems.

\emph{Structured information flow.} Downward flow provides
coordination signals through strategic guidance without requiring
centralized control at each timestep. Upward flow aggregates tactical
performance $R^{\text{tac}}_k = \sum_{t,i} r^{\text{tac}}_{i,t}$ into
strategic rewards, keeping high-level plans grounded in execution
realities. The strategic layer maintains exclusive access to
$s^{\text{env}}_k$, preserving the information asymmetry that enables
scalable distributed execution.

\emph{Distributed tactical learning.} Each tactical agent maintains
its own policy $\pi^{\text{tac}}_i$ and experience buffer
$\mathcal{D}^{\text{tac}}_i$, updating parameters using only local
gradients without inter-agent communication during learning, so that
tactical computation scales linearly with $N$.

\subsection{Algorithm Specification}

Algorithm~\ref{alg:ahpg} presents the complete AHPG procedure.
Strategic guidance is generated at the start of each new strategic
epoch (every $N_f$ episodes), and strategic parameters are updated at
the end of each epoch using the accumulated tactical performance.
Tactical parameters update every episode.

\begin{algorithm}[h]
  \caption{Asynchronous Hierarchical Policy Gradient (AHPG)}
  \label{alg:ahpg}
  \begin{algorithmic}[1]
    \REQUIRE Initial parameters $\theta^{\text{str}}_1$,
      $\{\theta^{\text{tac}}_{i,0}\}$
    \REQUIRE Hyperparameters $T$, $N_f$, $\alpha$, $\eta^{\text{tac}}$,
      $\lambda$, $c$
    \STATE Initialize episode counter $e \leftarrow 0$, strategic
      epoch $k \leftarrow 1$
    \FOR{episode $e = 1, 2, \ldots$}
      \STATE \textbf{// Strategic Planning Phase}
      \IF{$e \bmod N_f = 1$}
        \STATE Observe full state
          $s^{\text{str}}_k = (s^{\text{env}}_k, \mathbf{s}_k)$
        \STATE Generate
          $\mathbf{a}^{\text{str}}_k \sim
           \pi^{\text{str}}(s^{\text{str}}_k;\, \theta^{\text{str}}_k)$
        \STATE Broadcast $a^{\text{str}}_{i,k}$ to each agent $i$
      \ENDIF
      \STATE \textbf{// Tactical Execution Phase}
      \FOR{$t = 0$ to $T-1$}
        \FOR{agent $i \in \mathcal{V}$ in parallel}
          \STATE $o^{\text{tac}}_{i,t} =
            \phi_i(s_{i,t}, \{s_{j,t}\}_{j \in \mathcal{N}_i},
            a^{\text{str}}_{i,k})$
          \STATE $a^{\text{tac}}_{i,t} \sim
            \pi^{\text{tac}}_i(o^{\text{tac}}_{i,t};\,
            \theta^{\text{tac}}_{i,e})$
          \STATE Store $(o^{\text{tac}}_{i,t}, a^{\text{tac}}_{i,t},
            r^{\text{tac}}_{i,t})$ in $\mathcal{D}^{\text{tac}}_i$
        \ENDFOR
      \ENDFOR
      \STATE \textbf{// Tactical Policy Update}
      \FOR{agent $i \in \mathcal{V}$ in parallel}
        \STATE $g^{\text{tac}}_i =
          \nabla_{\theta^{\text{tac}}_i} J^{\text{tac}}_i$
          using $\mathcal{D}^{\text{tac}}_i$
        \STATE $\theta^{\text{tac}}_{i,e+1} =
          \theta^{\text{tac}}_{i,e} + \eta^{\text{tac}} g^{\text{tac}}_i$
      \ENDFOR
      \STATE \textbf{// Strategic Update}
      \IF{$e \bmod N_f = 0$}
        \STATE $R^{\text{tac}}_k =
          \sum_{t=kT}^{(k+1)T-1} \sum_{i=1}^N r^{\text{tac}}_{i,t}$
        \STATE $r^{\text{str}}_k =
          R^{\text{global}}(s^{\text{env}}_k, \mathbf{a}^{\text{str}}_k)
          + \lambda R^{\text{tac}}_k$
        \STATE $g^{\text{str}} =
          \nabla_{\theta^{\text{str}}} J^{\text{str}}$
          using $\mathcal{D}^{\text{str}}$
        \STATE $\theta^{\text{str}}_{k+1} = \theta^{\text{str}}_k
          + (\alpha/\sqrt{k})\, g^{\text{str}}$
        \STATE $k \leftarrow k + 1$
      \ENDIF
    \ENDFOR
  \end{algorithmic}
\end{algorithm}

During the $N_f$ episodes between strategic updates, tactical agents
operate in a stationary environment with fixed guidance, converging
toward near-optimal policies. This ensures strategic updates observe
stable tactical performance that accurately estimates the tactical
contribution to strategic rewards.

For learning components, the strategic layer employs centralized policy
gradient algorithms (e.g., PPO \cite{schulman2017proximal}) leveraging
complete state access, while tactical agents use distributed policy
gradient or actor-critic variants operating on local observations. The
framework naturally extends to continuous control through continuous
guidance vectors and to partially observable settings through recurrent
policies.

\section{CONVERGENCE ANALYSIS}

We establish convergence for AHPG under the additive approximation
$J^{\text{global}}(\theta) = J^{\text{str}}(\theta^{\text{str}})
+ \sum_{i=1}^N J^{\text{tac}}_i(\theta^{\text{tac}}_i)$.

\begin{assumption}
  \label{assum:reg}
  \begin{enumerate}
    \item \textbf{Smoothness}: $J^{\text{str}}$ and each
      $J^{\text{tac}}_i$ are $L$-smooth with $L$-Lipschitz gradients.
    \item \textbf{Boundedness}: $J^{\text{str}} \le J^{\text{str}*}$
      and $J^{\text{tac}}_i \le J^{\text{tac}*}_i$ for all $i$.
    \item \textbf{Bounded variance}: Stochastic gradient estimates
      satisfy
      $\mathbb{E}[\|g^{\text{str}}_k - \nabla J^{\text{str}}_k\|^2]
      \le \sigma^2_{\text{str}}$ and
      $\mathbb{E}[\|g^{\text{tac}}_{i,e}
      - \nabla J^{\text{tac}}_i\|^2] \le \sigma^2_{\text{tac}}$.
    \item \textbf{Biased strategic gradient}:
      $\mathbb{E}[g^{\text{str}}_k \mid \theta^{\text{str}}_k]
      = \nabla J^{\text{str}}(\theta^{\text{str}}_k) + b_k$,
      where $b_k$ captures the error from evaluating strategic
      gradients with non-converged tactical parameters.
  \end{enumerate}
\end{assumption}

\begin{assumption}
  \label{assum:tac}
  \begin{enumerate}
    \item \textbf{PL condition}: Each $J^{\text{tac}}_i$ satisfies the
      $\mu$-PL inequality:
      $\|\nabla J^{\text{tac}}_i(\theta)\|^2 \ge
      2\mu(J^{\text{tac}*}_i - J^{\text{tac}}_i(\theta))$.
    \item \textbf{Warm-start}: At the beginning of each strategic epoch
      $k$, tactical agents are initialized from the converged state of
      the previous epoch. The initial optimality gap satisfies
      $\mathbb{E}[J^{\text{tac}*}_i
      - J^{\text{tac}}_i(\theta^{\text{tac}}_{i,kN_f})] \le \delta_k$,
      where $\delta_k = \mathcal{O}(\eta^{\text{str}}_k)
      = \mathcal{O}(1/\sqrt{k})$, reflecting the magnitude of the
      strategic parameter change between successive epochs.
    \item \textbf{Round-robin updates}: Within each tactical phase,
      agents update sequentially with others' parameters held fixed,
      so that the PL condition applies to each per-agent landscape
      independently.
  \end{enumerate}
\end{assumption}

\begin{remark}
  The PL condition in Assumption~\ref{assum:tac} holds for log-linear
  (softmax) policies in tabular settings \cite{agarwal2020optimality} and
  locally for over-parameterized neural networks near stationary points.
  The round-robin assumption is a standard device to decouple per-agent
  landscapes; it can be relaxed to simultaneous updates at the cost of
  a more involved coupling analysis.
\end{remark}

\begin{assumption}[Coupling structure]
  \label{assum:coupling}
  The strategic gradient is $L_b$-Lipschitz in the tactical parameters:
  \[
    \|\nabla J^{\text{str}}(\theta^{\text{str}},
      \boldsymbol{\theta}^{\text{tac}})
     - \nabla J^{\text{str}}(\theta^{\text{str}},
      \boldsymbol{\theta}^{\text{tac}*})\|
    \le L_b \|\boldsymbol{\theta}^{\text{tac}}
               - \boldsymbol{\theta}^{\text{tac}*}\|.
  \]
\end{assumption}

\begin{theorem}[AHPG Convergence]
  \label{thm:convergence}
  Under Assumptions~\ref{assum:reg}--\ref{assum:coupling}, with
  $\eta^{\text{str}}_k = \alpha/\sqrt{k}$ where $\alpha \le 1/(8L)$,
  $\eta^{\text{tac}} = \beta/\sqrt{K}$ where $\beta < 1/(2L)$, and
  $N_f = c\sqrt{K}$ for a constant $c > 0$, after $K$ strategic
  updates AHPG achieves:
  \[
    \min_{k \in \{1,\ldots,K\}}
    \mathbb{E}[\|\nabla J^{\text{str}}(\theta^{\text{str}}_k)\|^2]
    = \mathcal{O}\!\left(\frac{\log K}{\sqrt{K}}\right),
  \]
  \[
    \frac{1}{KN_f}\sum_{k=1}^{K}\sum_{e=kN_f}^{(k+1)N_f-1}
    \sum_{i=1}^N
    \mathbb{E}[\|\nabla J^{\text{tac}}_i(\theta^{\text{tac}}_{i,e})\|^2]
    = \mathcal{O}\!\left(\frac{1}{\sqrt{K}}\right).
  \]
\end{theorem}

	\begin{proof}
		\textbf{Strategic layer.}
		By $L$-smoothness of $J^{\text{str}}$ and the update
		$\theta^{\text{str}}_{k+1} = \theta^{\text{str}}_k
		+ \eta^{\text{str}}_k g^{\text{str}}_k$:
		\begin{equation*}
		J^{\text{str}}_{k+1}
		\ge J^{\text{str}}_k
		+ \eta^{\text{str}}_k
		\langle\nabla J^{\text{str}}_k,\, g^{\text{str}}_k\rangle
		- \tfrac{L}{2}(\eta^{\text{str}}_k)^2
		\|g^{\text{str}}_k\|^2.
		\end{equation*}
		Taking expectation with
		$\mathbb{E}[g^{\text{str}}_k] = \nabla J^{\text{str}}_k + b_k$,
		we bound the two terms. For the inner product, applying Young's
		inequality $\langle a, b\rangle \ge
		-\tfrac{1}{4}\|a\|^2 - \|b\|^2$:
		\begin{equation*}
		\mathbb{E}[\langle\nabla J^{\text{str}}_k,
		g^{\text{str}}_k\rangle]
		\ge \tfrac{3}{4}\|\nabla J^{\text{str}}_k\|^2
		- \|b_k\|^2.
		\end{equation*}
		For the squared norm, using
		$\mathbb{E}[\|g^{\text{str}}_k\|^2]
		\le \|\nabla J^{\text{str}}_k + b_k\|^2 + \sigma^2_{\text{str}}
		\le 2\|\nabla J^{\text{str}}_k\|^2
		+ 2\|b_k\|^2 + \sigma^2_{\text{str}}$:
		\begin{align*}
		\mathbb{E}[J^{\text{str}}_{k+1}]
		&\ge J^{\text{str}}_k
		+ \!\left(\tfrac{3}{4}\eta^{\text{str}}_k
		- L(\eta^{\text{str}}_k)^2\right)
		\|\nabla J^{\text{str}}_k\|^2 \\
		&\quad
		- \!\left(\eta^{\text{str}}_k
		+ L(\eta^{\text{str}}_k)^2\right)\|b_k\|^2
		- \tfrac{L}{2}(\eta^{\text{str}}_k)^2
		\sigma^2_{\text{str}}.
		\end{align*}
		Since $\eta^{\text{str}}_k \le 1/(8L)$, we have
		$L(\eta^{\text{str}}_k)^2 \le \eta^{\text{str}}_k/8$, so
		$\tfrac{3}{4}\eta^{\text{str}}_k - L(\eta^{\text{str}}_k)^2
		\ge \tfrac{5}{8}\eta^{\text{str}}_k \ge \tfrac{1}{8}\eta^{\text{str}}_k$
		and $\eta^{\text{str}}_k + L(\eta^{\text{str}}_k)^2
		\le \tfrac{9}{8}\eta^{\text{str}}_k \le 2\eta^{\text{str}}_k$.
		Rearranging and telescoping over $k=1,\ldots,K$:
		\begin{align}
		\label{eq:str_telescope}
		\frac{1}{8}\sum_{k=1}^K \eta^{\text{str}}_k
		\mathbb{E}[\|\nabla J^{\text{str}}_k\|^2]
		&\le \Delta J^{\text{str}}
		+ 2\sum_{k=1}^K \eta^{\text{str}}_k \|b_k\|^2
		\nonumber\\
		&\quad
		+ \frac{L\sigma^2_{\text{str}}}{2}
		\sum_{k=1}^K (\eta^{\text{str}}_k)^2,
		\end{align}
		where $\Delta J^{\text{str}} = J^{\text{str}*} - J^{\text{str}}_1$.
		
		\textbf{Bias bound.}
		By Assumption~\ref{assum:coupling},
		$\|b_k\|^2 \le L_b^2 \|\boldsymbol{\theta}^{\text{tac}}_{(k+1)N_f}
		- \boldsymbol{\theta}^{\text{tac}*}\|^2$.
		Running $N_f = c\sqrt{K}$ SGD steps under the PL condition
		(Assumption~\ref{assum:tac}) with step size $\beta/\sqrt{K}$ and
		warm-start gap $\delta_k$ yields contraction factor
		$\rho = (1 - 2\mu\beta/\sqrt{K})^{c\sqrt{K}}
		\to e^{-2\mu\beta c} < 1$,
		so the residual tactical deviation satisfies
		$\|b_k\|^2 = \mathcal{O}(\delta_k + 1/\sqrt{K})
		= \mathcal{O}(1/\sqrt{k})$.
		
		\textbf{Completing the strategic bound.}
		With $\eta^{\text{str}}_k = \alpha/\sqrt{k}$:
		\begin{align*}
		\sum_{k=1}^K \eta^{\text{str}}_k\|b_k\|^2
		&= \mathcal{O}\!\left(\alpha\sum_{k=1}^K
		\tfrac{1}{\sqrt{k}} \cdot \tfrac{1}{\sqrt{k}}\right)
		= \mathcal{O}(\log K),\\
		\sum_{k=1}^K(\eta^{\text{str}}_k)^2
		&= \alpha^2\sum_{k=1}^K \tfrac{1}{k}
		= \mathcal{O}(\log K).
		\end{align*}
		Dividing~\eqref{eq:str_telescope} by
		$\sum_{k=1}^K \eta^{\text{str}}_k
		= \mathcal{O}(\alpha\sqrt{K})$ gives
		$\min_k \mathbb{E}[\|\nabla J^{\text{str}}_k\|^2]
		= \mathcal{O}(\log K/\sqrt{K})$.
		
		\textbf{Tactical layer.}
		Within epoch $k$, $\theta^{\text{str}}_k$ is fixed, so
		$J^{\text{tac}}_i$ is stationary. Standard SGD under the
		$\mu$-PL condition over $N_f$ steps gives per-epoch average
		gradient norm $\mathcal{O}(\delta_k + 1/\sqrt{K})$.
		Averaging over $K$ epochs and using
		$K^{-1}\sum_{k=1}^K\delta_k
		= K^{-1}\sum_{k=1}^K\mathcal{O}(k^{-1/2})
		= \mathcal{O}(K^{-1/2})$,
		the overall tactical average is $\mathcal{O}(1/\sqrt{K})$,
		completing the proof.
	\end{proof}

\begin{remark}
  Theorem~\ref{thm:convergence} is established for policy gradient
  implementations. The DQN-based experiments in Section~\ref{sec:exp}
  serve as an empirical proof of concept; extending the convergence
  analysis to Q-learning variants is left for future work.
\end{remark}

%%%%%%%%%%%%%%%%%%%%%%%%%%%%%%%%%%%%%%%%%%%%%%%%%%%%%%%%%%%%%%%%
\section{EXPERIMENTAL VALIDATION}
\label{sec:exp}

\subsection{Experiment Setting}

We evaluate CHMAS on a 25$\times$25 GridWorld environment with $N=4$
agents performing cooperative foraging, instantiating the key features
of hierarchical multi-agent coordination: information asymmetry,
temporal abstraction, and bidirectional coupling.

\begin{figure}[h]
  \centering
  \includegraphics[width=0.8\columnwidth]{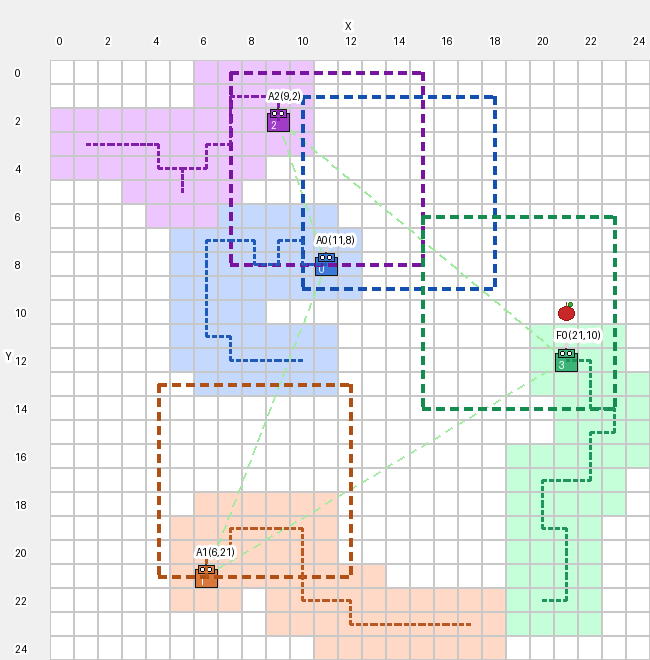}
  \caption{The 25$\times$25 GridWorld environment. Four agents
    (A0-blue, A1-orange, A2-purple, A3-green) operate within their
    assigned 9$\times$9 regions (dashed boxes) determined by strategic
    guidance. Colored areas indicate visited cells tracked in the
    global state $s^{\text{env}}$, while agents pursue local resources
    (red apple).}
  \label{fig:environment}
\end{figure}

The global environmental state $s^{\text{env}} \in \{0,1\}^{625}$
tracks cell visitation, representing system-wide coverage that
individual agents cannot observe. The strategic layer accesses the
complete state $s^{\text{str}} = (s^{\text{env}}, s_1, s_2, s_3, s_4)$, while each tactical agent $i$ observes only local
information $o^{\text{tac}}_i \in \mathbb{R}^{20}$ comprising its
position, nearby resources, and strategic guidance.

Strategic actions assign each agent a 9$\times$9 region center
$a^{\text{str}}_i = (g^x_i, g^y_i) \in [4,20]^2$ at the start of each
epoch ($T=5$ timesteps). Tactical agents select discrete actions
$a^{\text{tac}}_i \in \{0,1,2,3,4,5\}$ (movement, collection, wait).

The reward structure implements bidirectional coupling: tactical agents
receive $r^{\text{tac}}_{i,t} = r^{\text{collect}}_{i,t} - 0.01 -
0.01 \cdot \mathbf{1}[\text{outside region}]$, while the strategic
layer exclusively receives global coverage reward
$R^{\text{global}}_k = 0.01 \times |\text{newly visited cells}|$
computed from $s^{\text{env}}$. Strategic rewards integrate both:
$R^{\text{str}}_k = R^{\text{global}}_k + \lambda \sum_{t,i}
r^{\text{tac}}_{i,t}$ with $\lambda = 0.5$.

The implementation is publicly available at
\url{https://github.com/EricDmWang/Hierarchy_RL.git}.

\subsection{Results and Analysis}

We implement the tactical layer using Multi-Agent DQN (MA-DQN) with
parameter sharing, dueling architecture, double Q-learning, and
prioritized experience replay. The strategic layer employs a centralized
DQN on the global state. Training runs for 400 episodes with 5
independent seeds.

\begin{figure*}[h]
  \centering
  \begin{subfigure}{0.48\columnwidth}
    \includegraphics[width=\textwidth]{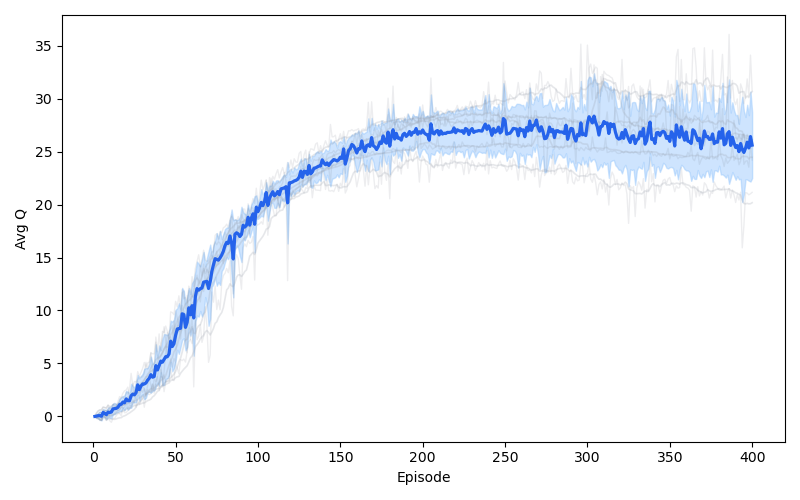}
    \caption{Tactical Q-values}
  \end{subfigure}
  \hfill
  \begin{subfigure}{0.48\columnwidth}
    \includegraphics[width=\textwidth]{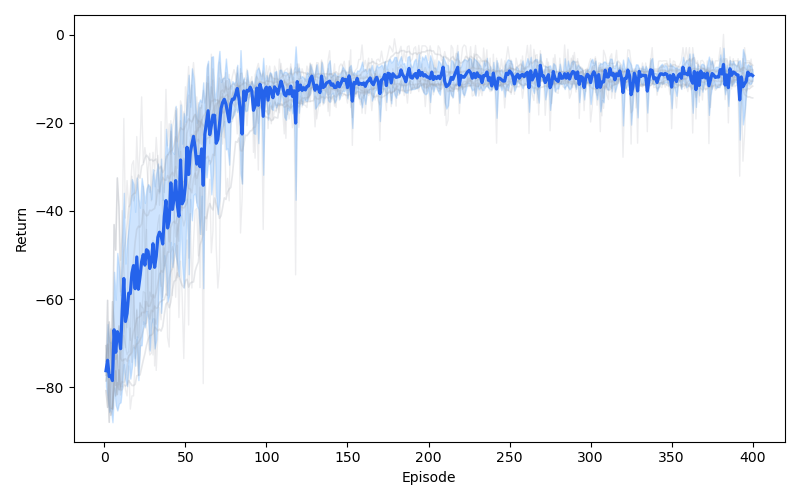}
    \caption{Tactical returns}
  \end{subfigure}
  \begin{subfigure}{0.48\columnwidth}
    \includegraphics[width=\textwidth]{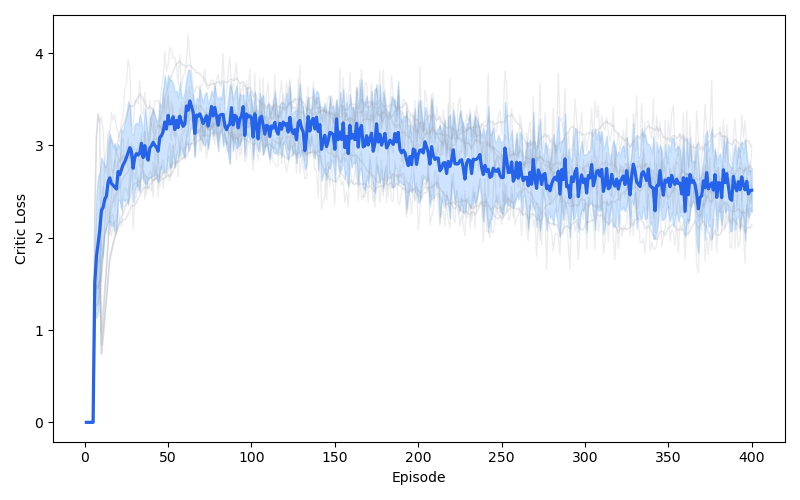}
    \caption{Tactical critic loss}
  \end{subfigure}
  \hfill
  \begin{subfigure}{0.48\columnwidth}
    \includegraphics[width=\textwidth]{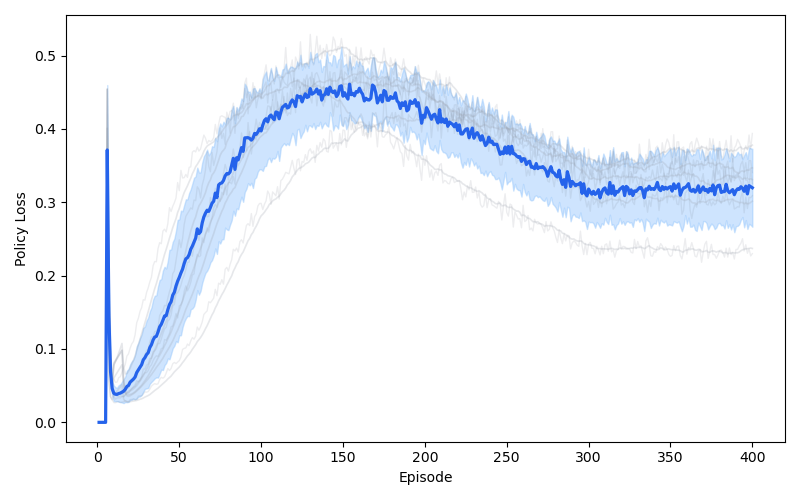}
    \caption{Tactical policy loss}
  \end{subfigure}
  \caption{Tactical layer learning dynamics.}
  \label{fig:tactical_results}
\end{figure*}

Figure~\ref{fig:tactical_results} shows the tactical layer's learning
progression. Q-values (a) grow rapidly before stabilizing near episode
300. Returns (b) improve from approximately $-80$ to $-10$, showing
agents learn efficient resource collection while respecting region
constraints. Critic loss (c) and policy loss (d) peak early during
exploration and then decrease as policies stabilize.

\begin{figure*}[h]
  \centering
  \begin{subfigure}{0.48\columnwidth}
    \includegraphics[width=\textwidth]{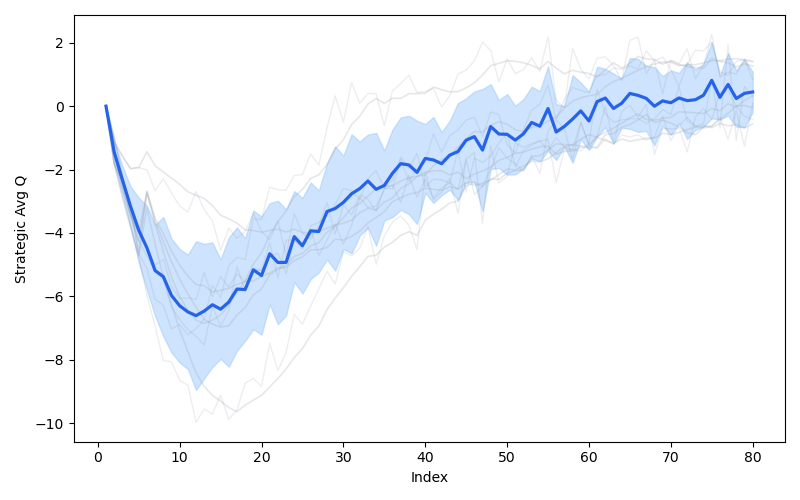}
    \caption{Strategic Q-values}
  \end{subfigure}
  \hfill
  \begin{subfigure}{0.48\columnwidth}
    \includegraphics[width=\textwidth]{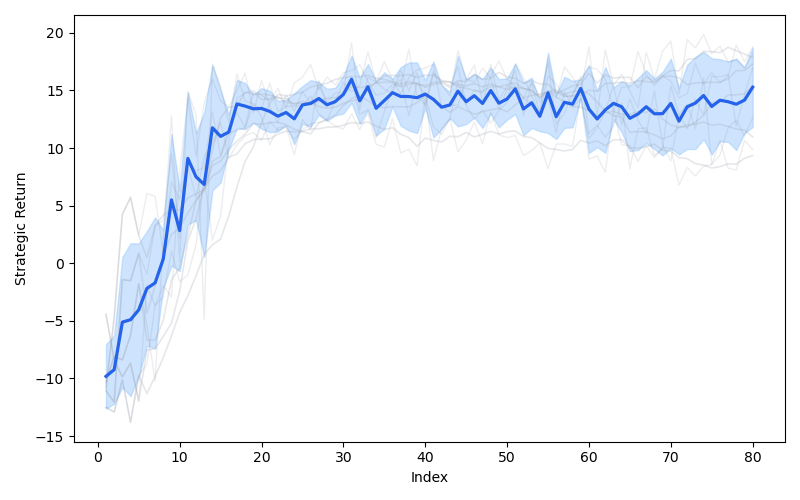}
    \caption{Strategic returns}
  \end{subfigure}
  \begin{subfigure}{0.48\columnwidth}
    \includegraphics[width=\textwidth]{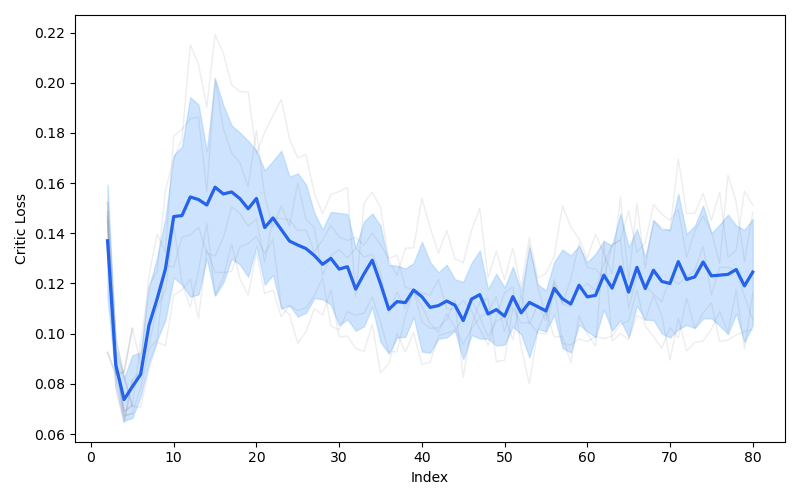}
    \caption{Strategic critic loss}
  \end{subfigure}
  \hfill
  \begin{subfigure}{0.48\columnwidth}
    \includegraphics[width=\textwidth]{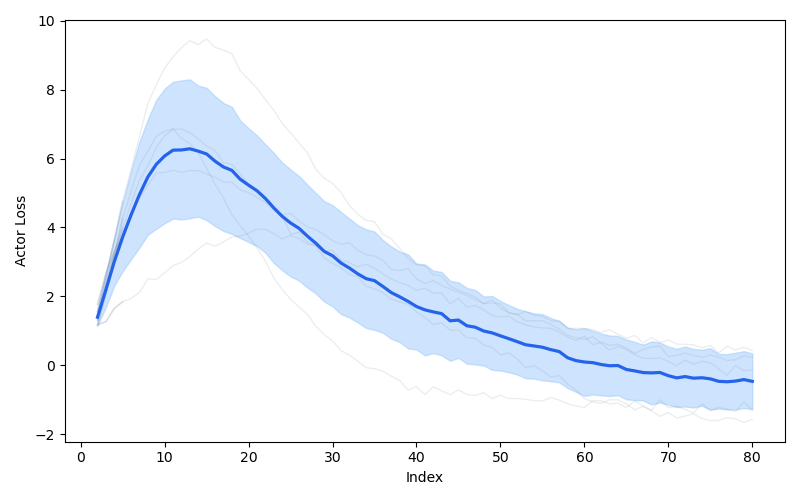}
    \caption{Strategic actor loss}
  \end{subfigure}
  \caption{Strategic layer learning dynamics with $N_f=5$.}
  \label{fig:strategic_results}
\end{figure*}

Figure~\ref{fig:strategic_results} shows the strategic layer's
asynchronous learning with updates every $N_f=5$ episodes. Strategic
Q-values (a) improve from $-8$ to near $0$, and strategic returns (b)
increase from $-10$ to approximately $15$, demonstrating successful
optimization of global coverage objectives. The delayed onset of
learning (first 10--20 episodes) reflects the asynchronous protocol
allowing tactical policies to stabilize before strategic updates
commence. Low variance across runs confirms algorithmic stability.

\subsection{Qualitative Analysis}

Figure~\ref{fig:execution} shows the execution trajectory of a trained
AHPG model over a complete episode. The strategic layer learns to assign
non-overlapping regions that collectively tile the grid, initially
placing agents in corners and then adapting assignments to unexplored
areas. Agents balance region compliance with opportunistic resource
collection. By frame (e), full resource collection with maintained
spatial coordination validates the effectiveness of bidirectional
hierarchical control.

\begin{figure*}[t]
  \centering
  \begin{subfigure}{0.19\textwidth}
    \includegraphics[width=\textwidth]{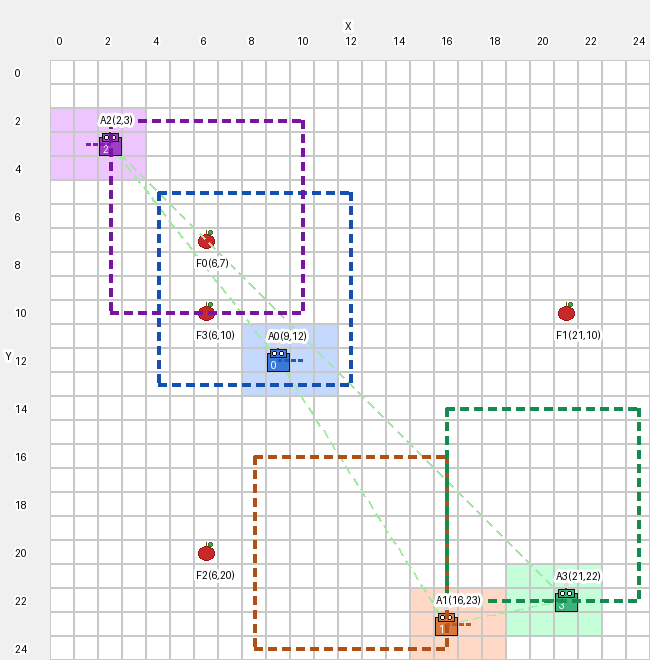}
    \caption{$t=0$}
  \end{subfigure}
  \hfill
  \begin{subfigure}{0.19\textwidth}
    \includegraphics[width=\textwidth]{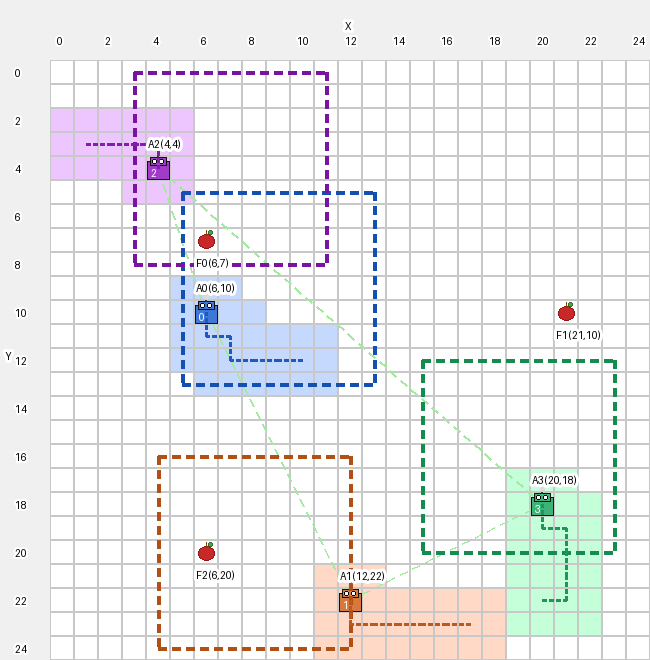}
    \caption{$t=6$}
  \end{subfigure}
  \hfill
  \begin{subfigure}{0.19\textwidth}
    \includegraphics[width=\textwidth]{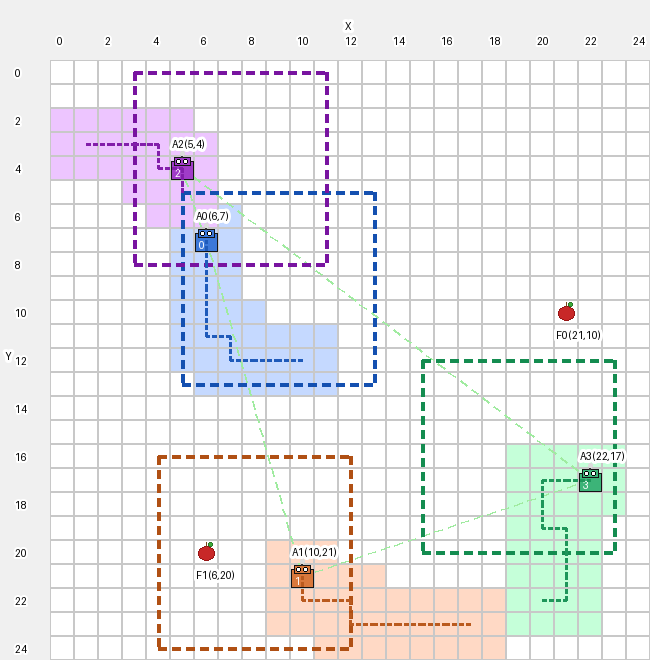}
    \caption{$t=9$}
  \end{subfigure}
  \hfill
  \begin{subfigure}{0.19\textwidth}
    \includegraphics[width=\textwidth]{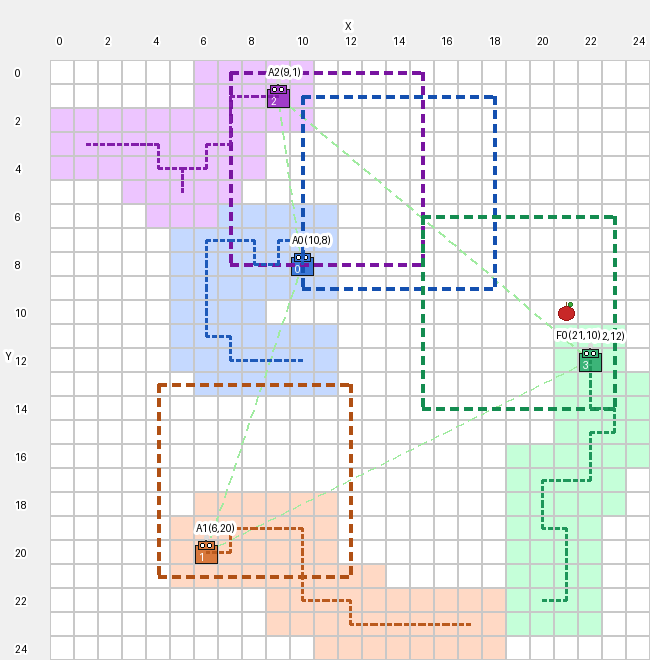}
    \caption{$t=16$}
  \end{subfigure}
  \hfill
  \begin{subfigure}{0.19\textwidth}
    \includegraphics[width=\textwidth]{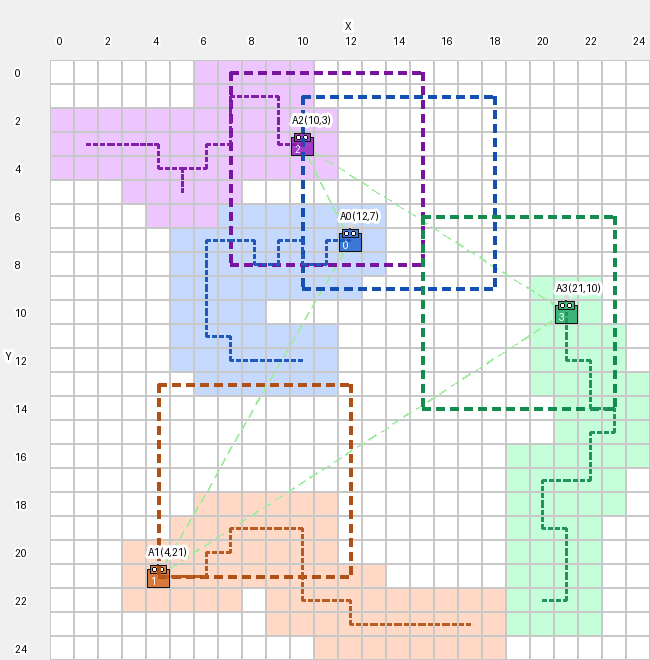}
    \caption{$t=19$}
  \end{subfigure}
  \caption{Evolution of hierarchical coordination during execution.
    Agents operate within their assigned 9$\times$9 regions (dashed
    boxes) while collectively achieving grid coverage. Red apples
    represent resources collected opportunistically.}
  \label{fig:execution}
\end{figure*}

%%%%%%%%%%%%%%%%%%%%%%%%%%%%%%%%%%%%%%%%%%%%%%%%%%%%%%%%%%%%%%%%
\section{CONCLUSION}

This paper introduced CHMAS, a hierarchical MARL framework based on
bidirectional coupling between strategic and tactical layers. By
enabling strategic planning with exclusive global state access while
maintaining distributed tactical execution, CHMAS provides a new
architectural paradigm for coordinating agents across multiple temporal
scales. The asynchronous update protocol, with a decaying strategic
step size schedule, manages the non-stationarity inherent in
hierarchical learning and admits $\mathcal{O}(\log K/\sqrt{K})$
convergence guarantees for the strategic layer.

The bidirectional coupling mechanism offers a principled approach to
constrained reinforcement learning: strategic guidance encodes global
constraints while tactical feedback ensures feasibility, with $\lambda$
providing tunable constraint satisfaction.

%Future work includes
%extending the convergence analysis to Q-learning implementations,
%adding empirical ablations on $\lambda$ and system scale, and applying
%the framework to safety-critical domains such as autonomous vehicle
%coordination.

\bibliographystyle{IEEEtran}
\bibliography{references}

\end{document}